\documentclass{article}
\usepackage{ijcai26}

\newdimen\proofrulebreadth \proofrulebreadth=.05em
\newdimen\proofdotseparation \proofdotseparation=1.25ex
\newdimen\proofrulebaseline \proofrulebaseline=2ex
\newcount\proofdotnumber \proofdotnumber=3
\let\then\relax
\def\hfi{\hskip0pt plus.0001fil}
\mathchardef\squigto="3A3B
\newif\ifinsideprooftree\insideprooftreefalse
\newif\ifonleftofproofrule\onleftofproofrulefalse
\newif\ifproofdots\proofdotsfalse
\newif\ifdoubleproof\doubleprooffalse
\let\wereinproofbit\relax
\newdimen\shortenproofleft
\newdimen\shortenproofright
\newdimen\proofbelowshift
\newbox\proofabove
\newbox\proofbelow
\newbox\proofrulename
\def\shiftproofbelow{\let\next\relax\afterassignment\setshiftproofbelow\dimen0 }
\def\shiftproofbelowneg{\def\next{\multiply\dimen0 by-1 }%
\afterassignment\setshiftproofbelow\dimen0 }
\def\setshiftproofbelow{\next\proofbelowshift=\dimen0 }
\def\setproofrulebreadth{\proofrulebreadth}

\def\prooftree{
\ifnum  \lastpenalty=1
\then   \unpenalty
\else   \onleftofproofrulefalse
\fi
\ifonleftofproofrule
\else   \ifinsideprooftree
        \then   \hskip.5em plus1fil
        \fi
\fi
\bgroup
\setbox\proofbelow=\hbox{}\setbox\proofrulename=\hbox{}%
\let\justifies\proofover\let\leadsto\proofoverdots\let\Justifies\proofoverdbl
\let\using\proofusing\let\[\prooftree
\ifinsideprooftree\let\]\endprooftree\fi
\proofdotsfalse\doubleprooffalse
\let\thickness\setproofrulebreadth
\let\shiftright\shiftproofbelow \let\shift\shiftproofbelow
\let\shiftleft\shiftproofbelowneg
\let\ifwasinsideprooftree\ifinsideprooftree
\insideprooftreetrue
\setbox\proofabove=\hbox\bgroup$\displaystyle 
\let\wereinproofbit\prooftree
\shortenproofleft=0pt \shortenproofright=0pt \proofbelowshift=0pt
\onleftofproofruletrue\penalty1
}

\def\eproofbit{
\ifx    \wereinproofbit\prooftree
\then   \ifcase \lastpenalty
        \then   \shortenproofright=0pt  
        \or     \unpenalty\hfil         
        \or     \unpenalty\unskip       
        \else   \shortenproofright=0pt  
        \fi
\fi
\global\dimen0=\shortenproofleft
\global\dimen1=\shortenproofright
\global\dimen2=\proofrulebreadth
\global\dimen3=\proofbelowshift
\global\dimen4=\proofdotseparation
\global\count255=\proofdotnumber
$\egroup  
\shortenproofleft=\dimen0
\shortenproofright=\dimen1
\proofrulebreadth=\dimen2
\proofbelowshift=\dimen3
\proofdotseparation=\dimen4
\proofdotnumber=\count255
}

\def\proofover{
\eproofbit 
\setbox\proofbelow=\hbox\bgroup 
\let\wereinproofbit\proofover
$\displaystyle
}%
\def\proofoverdbl{
\eproofbit 
\doubleprooftrue
\setbox\proofbelow=\hbox\bgroup 
\let\wereinproofbit\proofoverdbl
$\displaystyle
}%
\def\proofoverdots{
\eproofbit 
\proofdotstrue
\setbox\proofbelow=\hbox\bgroup 
\let\wereinproofbit\proofoverdots
$\displaystyle
}%
\def\proofusing{
\eproofbit 
\setbox\proofrulename=\hbox\bgroup 
\let\wereinproofbit\proofusing
\kern0.3em$
}

\def\endprooftree{
\eproofbit 
  \dimen5 =0pt
\dimen0=\wd\proofabove \advance\dimen0-\shortenproofleft
\advance\dimen0-\shortenproofright
\dimen1=.5\dimen0 \advance\dimen1-.5\wd\proofbelow
\dimen4=\dimen1
\advance\dimen1\proofbelowshift \advance\dimen4-\proofbelowshift
\ifdim  \dimen1<0pt
\then   \advance\shortenproofleft\dimen1
        \advance\dimen0-\dimen1
        \dimen1=0pt
        \ifdim  \shortenproofleft<0pt
        \then   \setbox\proofabove=\hbox{%
                        \kern-\shortenproofleft\unhbox\proofabove}%
                \shortenproofleft=0pt
        \fi
\fi
\ifdim  \dimen4<0pt
\then   \advance\shortenproofright\dimen4
        \advance\dimen0-\dimen4
        \dimen4=0pt
\fi
\ifdim  \shortenproofright<\wd\proofrulename
\then   \shortenproofright=\wd\proofrulename
\fi
\dimen2=\shortenproofleft \advance\dimen2 by\dimen1
\dimen3=\shortenproofright\advance\dimen3 by\dimen4
\ifproofdots
\then
        \dimen6=\shortenproofleft \advance\dimen6 .5\dimen0
        \setbox1=\vbox to\proofdotseparation{\vss\hbox{$\cdot$}\vss}%
        \setbox0=\hbox{%
                \advance\dimen6-.5\wd1
                \kern\dimen6
                $\vcenter to\proofdotnumber\proofdotseparation
                        {\leaders\box1\vfill}$%
                \unhbox\proofrulename}%
\else   \dimen6=\fontdimen22\the\textfont2 
        \dimen7=\dimen6
        \advance\dimen6by.5\proofrulebreadth
        \advance\dimen7by-.5\proofrulebreadth
        \setbox0=\hbox{%
                \kern\shortenproofleft
                \ifdoubleproof
                \then   \hbox to\dimen0{%
                        $\mathsurround0pt\mathord=\mkern-6mu%
                        \cleaders\hbox{$\mkern-2mu=\mkern-2mu$}\hfill
                        \mkern-6mu\mathord=$}%
                \else   \vrule height\dimen6 depth-\dimen7 width\dimen0
                \fi
                \unhbox\proofrulename}%
        \ht0=\dimen6 \dp0=-\dimen7
\fi
\let\doll\relax
\ifwasinsideprooftree
\then   \let\VBOX\vbox
\else   \ifmmode\else$\let\doll=$\fi
        \let\VBOX\vcenter
\fi
\VBOX   {\baselineskip\proofrulebaseline \lineskip.2ex
        \expandafter\lineskiplimit\ifproofdots0ex\else-0.6ex\fi
        \hbox   spread\dimen5   {\hfi\unhbox\proofabove\hfi}%
        \hbox{\box0}%
        \hbox   {\kern\dimen2 \box\proofbelow}}\doll%
\global\dimen2=\dimen2
\global\dimen3=\dimen3
\egroup 
\ifonleftofproofrule
\then   \shortenproofleft=\dimen2
\fi
\shortenproofright=\dimen3
\onleftofproofrulefalse
\ifinsideprooftree
\then   \hskip.5em plus 1fil \penalty2
\fi
}

\usepackage{times}
\usepackage{soul}
\usepackage{url}
\usepackage[colorlinks]{hyperref}
\usepackage[utf8]{inputenc}
\usepackage[small]{caption}
\usepackage{graphicx}
\usepackage{amsmath}
\usepackage{amsthm}
\usepackage{amsfonts} %
\usepackage{amssymb} %
\usepackage{booktabs}
\usepackage{algorithm}
\usepackage{algorithmic}
\usepackage[switch]{lineno}
\usepackage[most]{tcolorbox}

\usepackage{mdwlist}
\usepackage{wrapstuff}

\newtheorem{thrm}{Theorem}[section]

\newtheorem{lemm}[thrm]{Lemma}

\newtheorem{_nttn}[thrm]{Notation}
\newtheorem{_defn}[thrm]{Definition}
\newtheorem{_xmpl}[thrm]{Example}
\newtheorem{_rmrk}[thrm]{Remark}
\newenvironment{nttn}{\begin{_nttn}\normalfont}{\end{_nttn}}
\newenvironment{defn}{\begin{_defn}\normalfont}{\end{_defn}}
\newenvironment{xmpl}{\begin{_xmpl}\normalfont}{\end{_xmpl}}
\newenvironment{rmrk}{\begin{_rmrk}\normalfont}{\end{_rmrk}}

\newcommand\deffont[1]{{\bfseries #1}}
\newcommand\f[1]{{\text{$\mathit{#1}$}}}  %
\newcommand\tf[1]{{\text{$\mathsf{#1}$}}} %
\newcommand\pfs[1]{{\text{$\mathtt{#1}$}}} %

\newcommand\bbN{\mathbb N}

\newcommand\ment{\vDash}
\newcommand\nment{\nvDash}

\newcommand\fromto[2]{#1..#2}
\newcommand\afromto[1]{[#1]}
\newcommand\interpolant{\f{itplnt}}
\newcommand\theX{X}

\newcommand\thex{x}

\newcommand\thea{a}
\newcommand\theb{b}
\newcommand\thevft{t}

\newcommand\lmodel{[\hspace{-0.2em}[}
\newcommand\rmodel{]\hspace{-0.2em}]}
\newcommand\model[1]{{\lmodel #1 \rmodel}}
\newcommand\sem[1]{\model{#1}_\varsigma}

\DeclareMathSymbol{\shortminus}{\mathbin}{AMSa}{"39}
\newcommand\minus{{\shortminus}}
\newcommand\plus{{+}}
\newcommand\Forall[1]{\forall #1.}

\newcommand\arity{\f{arity}}

\newcommand\compressthis[1]{{\pmb{\hspace{.8pt}\raisebox{.5pt}{\scalebox{.85}{$#1$}}\hspace{.2pt}}}}
\newcommand\syntaxrel[1]{\mathbin{{\pmb{#1}}}}

\newcommand\teq{\syntaxrel{\text{=}}}
\newcommand\tleq{\syntaxrel{\leq}} %
\newcommand\tless{\syntaxrel{<}}

\newcommand\tand{\syntaxrel{\wedge}}
\newcommand\tor{\syntaxrel{\vee}}

\newcommand\tall{{\compressthis{\forall}}}
\newcommand\tallc[1]{\tall_{\hspace{-1pt}#1}}

\title{Cryptographic certificates of validity for trustworthy AI}

\author{
    Murdoch J. Gabbay 
    \affiliations
    Heriot-Watt University, Edinburgh, UK 
    \emails
    m.gabbay@hw.ac.uk
}

\begin{document}

\maketitle

\begin{abstract}
We propose cryptographic certificates of validity for agentic AI systems. 
The core idea is to formally specify a correctness or policy condition as a logical predicate, compile this predicate to a witness-checking problem over polynomial constraints, and use a succinct cryptographic proof system (and optionally zero-knowledge) to certify that the condition holds.

This offers a middle ground between formal verification of source code, and cryptographic authentication.
An agent's action can be accompanied by an independently checkable proof that it satisfies an agreed formal policy, without requiring the verifier to trust the agent or to re-execute computation. 

We outline the approach at a high level, give the core mathematical translation, relate the proposal to proof-carrying code, zkVMs, formal methods, and agent governance, and note the specification, auditing, and deployment questions that a full implementation must answer.

\end{abstract}

\section{Introduction}

Agentic AI systems are beginning to take actions, not merely produce recommendations.
A travel assistant may book a flight; an enterprise agent may approve an invoice; a software agent may deploy code; a robotic agent may act in the physical world.
For such systems a key question is: \emph{is this action authorised and policy compliant; is it correct?}

Cryptographic signatures traditionally authenticate \emph{origin}, by certifying that the sender of a message knew a certain key.
Logs can be used retrospectively to reconstruct events. 
However, cryptographic signatures, logs, and monitoring cannot prove that an action will satisfy correctness, safety, or compliance conditions \emph{in advance}.
In this paper, we propose a complementary mechanism: cryptographic certificates of \emph{validity} for agent actions, before actions are carried out.

The idea is simple.  
\begin{enumerate*}
\item
Express the relevant policy or correctness condition as a formal specification in an appropriate logic. 
\item
Compile the resulting finite witness-checking problem to cryptography-friendly polynomial constraints (as we will sketch).  
\item
Require the agent or an associated prover to attach a succinct cryptographic proof that the constraints are satisfied. 
(In zero-knowledge settings, the verifier may check the certificate without learning private witness data.)
\end{enumerate*}
The slogan is: \emph{do not trust an action because of its origin; trust it because it carries cryptographically checkable evidence of correctness.}

A note on scope may now be helpful.
The mathematics in Section~\ref{sect.slice.of.maths} is application-agnostic: it provides an arithmetisation of first‑order logic validity to judgements on polynomial constraints.
This may be useful wherever a verifier requires succinct cryptographic evidence that a stated condition holds. 
This builds a bridge from first-order specification to cryptographic certification.
The setting of agentic AI makes the assurance problem concrete and urgent --- autonomous action is a signal case in which origin-based authentication and post-hoc logging may not be fully adequate. 
But the mathematical architecture which we use to address this problem --- predicate $\to$ polynomial constraint $\to$ cryptographic certificate --- is general, and might be usefully specialised to any domain where untrusted parties must provide checkable evidence of correctness.

From this perspective, systems of AI agents are just heterogeneous distributed systems of agents who may be independently developed and may act according to their own priorities. 
They may also (to borrow a term from distributed systems) be \emph{Byzantine}, meaning that they do not satisfy the expected specification. 
Our proposal is to make selected actions proof-carrying, in a very particular sense: for their action to be accepted, they must produce a cryptographic certificate of validity of their correctness predicate.

Advanced AI assistants can plan and execute sequences of actions~\cite{gabriel2024ethics}.
Thus we need accountability, which requires evidence of where, why, how, and by whom agents are used~\cite{chan2024visibility}.  
This paper offers a complementary assurance: further to identity, logging, and risk-management practices such as those articulated in the NIST AI Risk Management Framework~\cite{nist2023airmf}, selected agent actions should carry compact certificates that a formal specification has been validated.
We can think of this as a cryptographic analogue of proof-carrying code~\cite{necula1997pcc}, instantiated using succinct arguments and polynomial arithmetisation~\cite{thaler2022proofs,gabbay2024arithcomp,gabbay2025arithmetisinglogic,garreta2025zinc}.

One final note, for the record: a certificate of validity for a formal specification does not in and of itself prove that the formal specification is the \emph{right} specification.
It proves, under the stated assumptions and to within appropriate soundness bounds of the underlying cryptographic back end, that the formal specification has been validated.

\section{A compact slice of maths}
\label{sect.slice.of.maths}

\begin{nttn}
Define $\afromto{n}=\{1,\dots,n\}\subseteq\mathbb N$.
If $R$ is a ring and $\theX$ is a formal \deffont{variable symbol} then $P\in R[\theX]$ is a formal sum of powers of $\theX$, which we call a \deffont{polynomial} (with coefficients in $R$).
Polynomials can be added and multiplied in the usual way.
For example, $(1+\theX) * (1 -\theX) = 1-\theX^2\in\mathbb Z[\theX]$. 
\end{nttn}

\begin{nttn}
\label{nttn.PQ}
If $P,Q\in\mathbb Z[\theX]$ then write $P(Q)\in\mathbb Z[\theX]$ for the result of \deffont{instantiating} $\theX$ to $Q$ in $P$.
For example if $P=\theX+1$ and $Q=2*\theX$ then 
$P(Q)=2*\theX+1$.
\end{nttn}

\begin{figure}
\begin{equation*}
\begin{array}{r@{\ }l} %
t::=&\theX \mid \thea\in\mathbb Z \mid t+t \mid t*t 
\mid \pfs{C}_i(t) \mid \tf{len}(\pfs{C}) 
\mid \tf{reify}(\phi)  
\\
\phi,\psi::=
&t\teq t \mid \phi\tand\phi \mid \phi\tor\phi \mid 
\tallc{\pfs{C}}\theX.\phi \mid
\pfs{C}_i\tless \thevft \mid \thevft\tless \pfs{C}_i
\end{array}
\end{equation*}
{\footnotesize Above, $i{\in}\afromto{\f{arity}(\pfs{C})}$. 
These grammars are in BNF style, such that $\phi$ and $t$ to the right of $::=$ represent any predicate or term.
Thus $\phi\tand\phi$ denotes a conjunction of two (possibly non-equal) predicates, and $t\teq t$ denotes two (possibly non-equal) terms.}
\caption{Syntax of terms and predicates (Definition~\ref{defn.logical.syntax})}
\label{fig.syntax}
\end{figure}

\begin{defn}
\label{defn.logical.syntax}
Fix an \deffont{index variable symbol} $\theX$ and a set of \deffont{polynomial function symbols} $\pfs{C},\pfs{D}\in\tf{PolyFunc}$, each of which has a fixed but arbitrary \deffont{arity} $\f{arity}(\pfs{C})\in\bbN_{\geq 1}$.
Write $\pfs{C}:n$ for the assertion ``$\pfs{C}\in\tf{PolyFunc}$ and $\f{arity}(\pfs{C})=n$''.

Define \deffont{terms} and \deffont{predicates} inductively as in Figure~\ref{fig.syntax}. 
\end{defn}

\begin{nttn}
\label{nttn.abbrv}
We may write 
$\thevft\tleq\pfs{C}_i\tleq\thevft'$ for $(\thevft-1\tless\pfs{C}_i)\tand (\pfs{C}_i\tless \thevft'+1)$.
We may call predicates based on $\tless$ \deffont{range-checks}.
\end{nttn}

\begin{figure}
$$
\begin{array}{r@{\ }l@{\ \ }r@{\ }l}
\model{\theX}_\varsigma=&X
&
\model{\thea}_\varsigma=&\thea 
\\
\model{t+t'}_\varsigma=&\model{t}_\varsigma+\model{t'}_\varsigma
&
\model{t*t'}_\varsigma=&\model{t}_\varsigma*\model{t'}_\varsigma
\\
\model{\pfs{C}_i(t)}_\varsigma=&\interpolant(\varsigma(\pfs{C})_i)(\model{t}_\varsigma)
&
\model{\tf{len}(\pfs{C})}_\varsigma=&\f{len}(\varsigma(\pfs{C}))
\\
\model{\tf{reify}(\phi)}_\varsigma=&\model{\phi}_\varsigma
\\[2ex]
\model{t\teq t'}_\varsigma=&(\model{t}_\varsigma-\model{t'}_\varsigma)^2
\\
\model{\phi\tand\phi'}_\varsigma=&\model{\phi}_\varsigma+\model{\phi'}_\varsigma
&
\model{\phi\tor\phi'}_\varsigma=&\model{\phi}_\varsigma*\model{\phi'}_\varsigma
\\
\model{\tallc{\pfs{C}}\theX.\phi}_\varsigma=&\rlap{$\sum_{x{\in}\afromto{\f{len}(\varsigma(\pfs{C}))}} \model{\phi}_\varsigma(x)$}
\\[2ex]
\model{\pfs{C}_i < \thevft}_\varsigma=& 
\rlap{$\begin{cases}
0 &\Forall{j{\in}\afromto{\f{len}(\varsigma(\pfs{C}))}}(\varsigma(\pfs{C})_{i,j}<\model{\thevft}_\varsigma(j))
\\
1 &\text{otherwise}
\end{cases}$}
\\[2ex]
\model{\thevft < \pfs{C}_i}_\varsigma=& 
\rlap{$\begin{cases}
0 &\Forall{j{\in}\afromto{\f{len}(\varsigma(\pfs{C}))}}(\model{\thevft}_\varsigma(j)<\varsigma(\pfs{C})_{i,j})
\\
1 &\text{otherwise}
\end{cases}$}
\end{array}
$$
{\small
In clauses mentioned in $\pfs{C}_i$ above, we insist on $i\in\afromto{\f{arity}(\pfs{C})}$ for well-formedness.}
\caption{Polynomial semantics for terms and predicates (Definition~\ref{defn.the.denotation})}
\label{fig.polynomial.semantics}
\end{figure}

\begin{nttn}
\label{nttn.matrix.nttn}
An \deffont{integer matrix} is an element $C\in\mathbb Z^{\f{ar}\times \f{ln}}$ where $\f{ar}\in\mathbb N$ (for \emph{arity}) is the number of \deffont{rows} and $\f{ln}\in\mathbb N$ (for \emph{length}) is the number of \deffont{columns}. 
Write $\f{len}(C)$ for $\f{ln}$ the number of columns in $C$, which we may call the \deffont{length} of $C$.
Write $C_i$ for the $i$th row in $C$ (a vector of length $\f{len}(C)$).
Write $C_{i,j}$ for the $i,j$-th element in $C$; $i$ is the row number and $j$ is the column number.
\end{nttn}

\begin{defn}
\label{defn.interpolation}
Suppose $\vec \thea=(\thea_1,\dots,\thea_n)\in\mathbb Z^{\f{ln}}$ is a $\f{ln}$-tuple of integers, for $\f{ln}\in\mathbb N$.
Say that $P\in\mathbb Q[\theX]$ \deffont{interpolates} $\vec \thea$ at $(1,2,\dots,\f{ln})$ when $\Forall{\thex\in\afromto{\f{ln}}}P(\thex)=\thea_x$. 

It is a fact~\cite[Chapter~9]{Salgado_Wise_2022} that every $\vec \thea\in\mathbb Z^{\f{ln}}$ can be interpolated by a unique \deffont{polynomial interpolant}  $\interpolant(\thea_1,\dots,\thea_{\f{ln}})\in\mathbb Q[\theX]$ of degree at most $\f{ln}\minus 1$, so $\interpolant(\thea_1,\dots,\thea_{\f{ln}})(\thex)=\thea_\thex$.\footnote{The \href{https://en.wikipedia.org/wiki/Polynomial_interpolation}{Wikipedia page} and \href{https://math.libretexts.org/Bookshelves/Applied_Mathematics/Numerical_Methods_(Chasnov)/05\%3A_Interpolation/5.01\%3A_Polynomial_Interpolation}{this overview} are excellent and very accessible introductions.}
\end{defn}

\begin{defn}
\label{defn.the.denotation}
\leavevmode
\begin{enumerate*}
\item\label{item.denotation.valuation}
An \deffont{interpretation} $\varsigma$ assigns to $\pfs{C}\in\tf{PolyFunc}$ %
a matrix $\varsigma(\pfs{C})\in\mathbb Z^{\f{arity}(\pfs{C})\times \f{len}(\varsigma(\pfs{C}))}$.
$\varsigma(\pfs{C})$ always has $\f{arity}(\pfs{C})$ rows; the number of columns depends on $\varsigma$.
\item\label{item.poly.denotation}
Define \deffont{semantics} $\model{t}_\varsigma,\model{\phi}_\varsigma\in\mathbb Q[X]$ as in Figure~\ref{fig.polynomial.semantics}.
\item\label{item.judgement.denotation}
If $x\in\mathbb Z$ then define the \deffont{(validity) judgement} $x\ment_\varsigma\phi$ to mean $\model{\phi}_\varsigma(x)=0$.
\\
If $\phi$ is closed (has no free variables) then define the \deffont{(validity) judgement} $\ment_\varsigma\phi$ to mean
$\model{\phi}_\varsigma(0)=0$.\footnote{The choice of evaluation point $0$ is unimportant; we could equivalently use $\model{\phi}_\varsigma(1)$, since for closed $\phi$, $\Forall{x,x'}\model{\phi}_\varsigma(x)=\model{\phi}_\varsigma(x')$.}
\end{enumerate*}
\end{defn}

\begin{lemm}
\label{lemm.geq.0}
$\model{\phi}_\varsigma(x)\geq 0$ for every $x\in\mathbb Q$, and in particular for every $x\in\mathbb Z$.
\end{lemm}
\begin{proof}
Routine induction on Figure~\ref{fig.polynomial.semantics},
noting that: for $x,y\in\mathbb Q$,\ $(x-y)^2\geq 0$; and 
for $x,y\geq 0$,\ $x+y,x*y\geq 0$.
\end{proof}

We can now prove \emph{soundness} and \emph{completeness}; validity with respect to the predicate connectives in Figure~\ref{fig.syntax} behaves exactly as those symbols suggest:
\begin{thrm}[\cite{gabbay2024arithcomp,gabbay2025arithmetisinglogic}]
\label{thrm.basic.obs}
\leavevmode
\begin{enumerate*}
\item
$\thex\ment_\varsigma t\teq t'$ if and only if $\model{t}_\varsigma(\thex)=\model{t'}_\varsigma(\thex)$.  
\item\label{item.basic.obs.tand}
$\thex\ment_\varsigma\phi\tand\phi'$ if and only if $\thex\ment_\varsigma\phi\ \land\ \thex\ment_\varsigma\phi'$ and  
$\thex\ment_\varsigma\phi\tor\phi'$ if and only if $\thex\ment_\varsigma\phi\ \lor\ \thex\ment_\varsigma\phi'$.
\item\label{item.basic.obs.tall}
$\thex\ment_\varsigma\tallc{\pfs{C}}\theX.\phi$ if and only if $\thex'\ment_\varsigma\phi$ for every $\thex'\in\afromto{\f{len}(\varsigma(\pfs{C}))}$.
\item
$\thex\ment_\varsigma \thevft\tless\pfs{C}_i$ if and only if $\varsigma(\pfs{C})_{i,j}$ is strictly greater than $\model{\thevft}_\varsigma(j)$ for every $j\in\afromto{\f{len}(\varsigma(\pfs{C}))}$, and similarly for
$\thex\ment_\varsigma \pfs{C}_i\tless \thevft$.
\end{enumerate*}
\end{thrm}
\begin{proof}
Routine from Figure~\ref{fig.polynomial.semantics} and Lemma~\ref{lemm.geq.0}, noting that: for $x,y\in\mathbb Q$,\ $(x\minus y)^2=0$ iff $x=y$; and for $x,y\geq 0$,\ $x+y=0$ iff $x{=}0\land y{=}0$ and $x*y=0$ iff $x{=}0\lor y{=}0$.
\end{proof}

\section{Examples}
\label{sect.examples}

\subsection{Three one-line examples}

Section~\ref{sect.slice.of.maths} takes up one page of dense mathematics and concludes with one Theorem.
This is compact and theoretical much as logic gates are compact and theoretical: by composing together simple components we attain great expressivity, and by suitable implementation we can accomplish remarkable things.
We give the reader a taste of how this works.

\begin{xmpl}
\label{xmpl.simplest.example.semantics}
We can express truth as $0\teq 0$ and false as $0\teq 1$.
The reader can check that $\model{0\teq 0}_\varsigma=0$ (thus $\ment_\varsigma 0\teq 0$) and $\model{0\teq 1}_\varsigma=1$ (thus $\nment_\varsigma 0\teq 1$) so these predicates do behave like `true' and `false'.

Assume $\pfs{m}:2$,\ $\pfs{a}:3$,\ and $\pfs{p}:1$ with \emph{validity predicates}
$$
\begin{array}{r@{\ }l}
\chi_{\pfs{m}}=&\tallc{\pfs{m}}\theX.(\pfs{m}_2(\theX)\,\teq\,(\minus 1)*\pfs{m}_1(\theX))
\\
\chi_{\pfs{a}}=&\tallc{\pfs{a}}\theX.(\pfs{a}_3(\theX)\,\teq\,\pfs{a}_1(\theX)+\pfs{a}_2(\theX))
\\
\chi_{\pfs{p}}=&(0\tless \pfs{p}_1).
\end{array}
$$
Given an interpretation $\varsigma$, we have
$$
\hspace{-1ex}\scalebox{0.82}{$
\begin{array}{r@{\ }l}
\model{\chi_{\pfs{m}}}_\varsigma 
=&
\sum_{x\in\afromto{\f{len}(\varsigma(\pfs{m}))}} 
(\interpolant(\pfs{m}_2)(x)- (\minus 1)*\interpolant(\pfs{m}_1)(x))^2
\\
\model{\chi_{\pfs{a}}}_\varsigma 
=&
\sum_{x\in\afromto{\f{len}(\varsigma(\pfs{a}))}} 
(\interpolant(\pfs{a}_3)(x)- (\interpolant(\pfs{a}_1)(x)\plus\interpolant(\pfs{a}_2)(x)))^2
\\
\model{\chi_{\pfs{p}}}_\varsigma
=&0\ \text{if}\ \Forall{x\in\afromto{\f{len}(\varsigma(\pfs{p}))}}\varsigma(\pfs{p})_{1,x}>0,\ \text{and}\ 1\text{ otherwise}
\end{array}
$}
$$
So:
\begin{itemize*}
\item
$\ment_\varsigma\chi_{\pfs{m}}$ when $\varsigma(\pfs{m})$ is a two-row matrix where in each column the second entry is the negation of the first.
Thus, each column in $\varsigma(\pfs{m})$ represents a pair $(x,\minus x)$ for some $x\in\mathbb Z$.

$\arity(\pfs{m})=2$ because $\lambda x.\minus x$ has one input and output.
\item
$\ment_\varsigma\chi_{\pfs{a}}$ when $\varsigma(\pfs{a})$ is a three-row matrix where in each column the third entry is the sum of the first and the second; thus $\varsigma(\pfs{a})$ has columns of the form $(x,x',x\plus x')$ and samples the graph of $\lambda x,x'.x\plus x'$ at $\f{len}(\varsigma(\pfs{a}))$ many points.

$\arity(\pfs{a})=3$ because $\lambda x,x'.x\plus x'$ has two inputs and one output.
\item
$\ment_\varsigma\chi_{\pfs{p}}$ when $\varsigma(\pfs{p})$ is a one-row matrix (a vector) of strictly positive numbers.

$\arity(\pfs{p})=1$ because `is strictly positive' is a predicate on one argument. 
\end{itemize*}
\end{xmpl}

\begin{figure}[t]
\centering
\begin{tabular}{@{}p{0.11\textwidth}p{0.3\textwidth}@{}}
\toprule
\textbf{Logic} & \textbf{Polynomial} 
\\
\midrule
Truth values & \(0=\) true/valid/success; strictly positive \(=\) false/invalid/failure. \\
Equality & \(\sem{t\teq t'}=(\sem t-\sem{t'})^2\), so equality is valid exactly at zero. \\
Conjunction & \(\sem{\phi\tand\psi}=\sem\phi+\sem\psi\); sum of nonnegatives is zero iff both are zero. \\
Disjunction & \(\sem{\phi\tor\psi}=\sem\phi*\sem\psi\); product is zero iff at least one factor is zero. \\
Range check & \(0\) when satisfied, \(1\) otherwise. 
\\
Top-level $\forall$ & Root test for a range of values.
\\
$\phi$ valid & $\sem{\phi}=0$
\\
$\phi$ invalid & $\sem{\phi}>0$ %
\\
\bottomrule
\end{tabular}
\caption{Table of semantic meaning (Remark~\ref{rmrk.semantic.meaning})} %
\label{fig.zero-true-positive-false}
\end{figure}

\begin{rmrk}
\label{rmrk.semantic.meaning}
A table of semantic meanings is in Figure~\ref{fig.zero-true-positive-false}.

The reader may be more used to seeing `true = 1; false = 0', but representing truth with $0$ and non-truth with strictly positive values also has pedigree.  
In a paper from 1943~\cite{kleene:recpq} on page~51 just after equation~17, Kleene writes \emph{``a representing function $\pi$ of $P$, the value of which is to be $0$, $1$, or undefined according as the value of $P$ is true, false, or undefined''}.  
More recently, in the \href{https://github.com/openbsd/src/blob/66b892aa8a957b4f3b178c2b7d0f92797065231e/include/sysexits.h\#L97}{\tt sysexits.h} file (\href{https://man.openbsd.org/sysexits}{credited to Eric Allman} in 1980) which describes status codes for system programs, $0$ represents successful termination and strictly positive values represent various kinds of failure.
\end{rmrk}

We will consider one more example, because it is an important one.
We will indicate that our framework can handle \emph{recursion}; this is the essence of looping and of Turing-complete computation.
The reader will have to excuse that our example will be extremely simple, but it is the \emph{structure} of the computation that interests us.

\subsection{A recursive example}

\begin{defn}
\label{defn.pow.clauses}
Recall the inductive definition of $a^b$ for $a,b\in\bbN_{\geq 0}$:
\begin{equation*}
\begin{array}{l@{\qquad}r@{\ }l}
\text{base case (BC)}
&
\f{pow}(\thea,0)&=1
\\
\text{inductive step (IS)}
&
\f{pow}(\thea,\theb\plus 1)&=\thea*\f{pow}(\thea,\theb)
\end{array}
\end{equation*}
\end{defn}

\begin{rmrk}
\label{rmrk.derivation-tree}
Note how Definition~\ref{defn.pow.clauses}, and definitions like it, implicitly equate \emph{computation}, \emph{inductively-defined relations}, and \emph{derivation}, because $\f{pow}$ is a relation, it is inductively defined as above, and when we try to compute a value --- $\f{pow}(2,2)$, say --- we generate a derivation-tree as follows:
\begin{equation*}
\begin{prooftree}
\[
\[
\justifies
\f{pow}(2,0)=1
\using{\text{base case}}
\]
\justifies
\f{pow}(2,1)=2
\using{\text{inductive step}}
\]
\justifies
\f{pow}(2,2)=4
\using{\text{inductive step}}
\end{prooftree}
\end{equation*}
This example, simple as it is, illustrates the close link between logic, computation, and inductively defined relations. 
\end{rmrk}

\begin{defn}
\label{defn.matrix.encodes.pow}
Suppose $\f{ln}\in\bbN_{\geq 1}$.
Say a $4\times \f{ln}$ matrix $C$ \deffont{encodes a partial power function} when
\begin{equation*}
\Forall{\thea\in\fromto{1}{\f{ln}}}\ C_{3,\thea}=C_{1,\thea}^{C_{2,\thea}} .
\end{equation*}
In words, $C$ encodes a partial power function when for each column, the third entry is equal to the first entry raised to the power of the second.
\end{defn} 

\begin{defn}
\label{defn.pow.const}
\label{defn.pow.poly.spec}
Assume %
$\pfs{pow}:4$, so $\varsigma(\pfs{pow})\in\mathbb Z^{4\times\f{len}(\varsigma(\pfs{pow}))}$. 
Figure~\ref{fig.pow.poly} expresses a validity predicate $\chi_{\pfs{pow}}$, that $\pfs{pow}$ encodes a partial power function such that:  
\begin{enumerate*}
\item
Row~1 (values of $\pfs{pow}_1(\thea)$) stores input values $a$.
\item
Row~2 (values of $\pfs{pow}_2(\thea)$) stores input values $b$.
\item
Row~3 (values of $\pfs{pow}_3(\thea)$) stores output values $a^b$.
\item
Row~4 (values of $\pfs{pow}_4(\thea)$) stores the index of a column representing the recursive call %
in the inductive step.
\end{enumerate*}
\end{defn}

\begin{figure}[t]
{\small
Write 
$\pfs{pow}_1(t){=}\tf{in}_1(t)$, 
$\pfs{pow}_2(t){=}\tf{in}_2(t)$, 
$\pfs{pow}_3(t){=}\tf{out}(t)$, and 
$\pfs{pow}_4(t){=}\tf{rec}(t)$.
\begin{equation*}
\begin{array}{r@{\ }l}
\tf{BC}(\theX) 
=& 
(\tf{in}_2(\theX)\teq 0) \tand (\tf{out}(\theX)\teq 1) 
\\
\tf{IS}(\theX)
=&
\tf{in}_1(\theX) \teq \tf{in}_1(\tf{rec}(\theX))
 \tand 
\tf{in}_2(\theX) \teq \tf{in}_2(\tf{rec}(\theX))\plus 1
 \tand 
\\
&\quad\tf{out}(\theX) \teq \tf{in}_1(\theX)*\tf{out}(\tf{rec}(\theX)) 
\\
\chi_{\pfs{pow}} =& 1\tleq\tf{rec}\tleq \tf{len}(\pfs{pow})\ \tand\ \tall_{\pfs{pow}}\theX. (\tf{BC}(\theX) \ \tor\ \tf{IS}(\theX)) 
\end{array}
\end{equation*}
}
\caption{Exponentiation $a^b$ (Definitions~\ref{defn.pow.clauses} and~\ref{defn.pow.poly.spec})}
\label{fig.pow.poly}
\end{figure}

\begin{lemm}
\label{lemm.pow.poly.spec}
If \ $\ment_\varsigma \chi_{\pfs{pow}}$  
then $\varsigma(\pfs{pow})$ encodes a partial power function. %
\end{lemm}
\begin{proof}
We examine the clauses in Figure~\ref{fig.pow.poly} and using Theorem~\ref{thrm.basic.obs} we check that they correctly encode the definition in Definition~\ref{defn.pow.clauses}.
\end{proof}

\begin{xmpl}
\label{xmpl.pow.do}
We illustrate some values for $\varsigma(\pfs{pow})$ such that $\ment_\varsigma\chi_{\pfs{pow}}$, so that by Lemma~\ref{lemm.pow.poly.spec} they encode a partial power function: 
\begin{equation*}
\scalebox{0.85}{$
\begin{pmatrix}
2 & 2 & 2 
\\ 
0 & 1 & 2 
\\ 
1 & 2 & 4 
\\ 
1 & 1 & 2 
\end{pmatrix} 
\quad
\begin{pmatrix}
2 & 2 & 2
\\ 
2 & 1 & 0
\\ 
4 & 2 & 1 
\\ 
2 & 3 & 3
\end{pmatrix} 
\quad
\begin{pmatrix}
2 & 2 & 2 & 2
\\ 
0 & 2 & 1 & 0
\\ 
1 & 4 & 2 & 1 
\\ 
2 & 3 & 1 & 1
\end{pmatrix} 
\quad
\begin{pmatrix}
3 & 2 & 3 & 2
\\ 
0 & 0 & 1 & 1
\\ 
1 & 1 & 3 & 2
\\ 
1 & 1 & 1 & 2
\end{pmatrix} 
$}
\end{equation*}
These all represent the derivation-tree from Remark~\ref{rmrk.derivation-tree}:
\begin{enumerate*}
\item
The leftmost matrix encodes in columns~1-3 respectively, the base case of a derivation that $2^0=1$, the step $2^1=2$, and the final step $2^2=4$.
Entry~4 in columns~2 and~3 contain a number which is a pointer to the previous step/column in the derivation; entry~4 in column~1 also contains a number but it is not used.
\item
The columns need not appear in a particular order, provided that the pointers match up; the next matrix illustrates this by putting the columns in the reverse order.
\item
Columns can be duplicated as per the third matrix.
\item
The fourth (rightmost) matrix includes information from computing $2^1=2$ and $3^1=3$.
\end{enumerate*}
\end{xmpl}

\begin{xmpl}
\label{xmpl.pow}
The converse implication for Lemma~\ref{lemm.pow.poly.spec} does not hold.
The predicate $\chi_{\pfs{pow}}$ certifies the pointwise correctness of the input-output samples in rows 1-3 \emph{and also} that those samples can be assembled into a well-formed derivation tree, whose recursive calls are represented by in-range pointers in row~4.

\begin{wrapstuff}[r]   
\scalebox{0.9}{$
\begin{pmatrix}
2\\
2\\ 
4\\
2
\end{pmatrix}
$}
\end{wrapstuff}
A matrix may encode true values of the partial power function while failing $\chi_{\pfs{pow}}$, if the certificate of derivation is missing or malformed.
For example, the one-column matrix to the right
contains the correct input-output sample (2,2,4), since \(2^2=4\).
However, it omits the recursive premise (2,1,2) required by the inductive step. 
Formally, \(\tf{len}(\pfs{pow})=1\) and the row~4 entry is 2, so the range check \(1\tleq \tf{rec} \tleq \tf{len}(\pfs{pow})\) fails.
Equivalently: the supposed recursive pointer points outside the matrix.

\begin{wrapstuff}[r]   
\scalebox{0.9}{$\begin{pmatrix}
2 & 2 & 2\\
2 & 1 & 0\\
4 & 2 & 1\\
2 & 3 & \underline{0}
\end{pmatrix}$}
\end{wrapstuff}
A second instructive failure is as per the three-column matrix to the right.
Rows 1-3 contain correct samples: \(2^2=4\), \(2^1=2\), and \(2^0=1\), and the first two recursive pointers are meaningful: column~1 points to column~2, and column~2 points to column~3.  
However, the underlined pointer in the base-case column is not in the range $\{1,2,3\}$.  
The base-case clause (BC) does not use this pointer, but our simple presentation imposes the uniform range check on all entries of the pointer row. 
Hence this matrix encodes correct samples of the partial power function but is not a well-formed
$\chi_{\pfs{pow}}$-certificate.  
Replacing the underlined 0 with any in-range value, for example with 1, would repair this.

In summary: rows~1-3 describe claimed values of the relation; row~4 supplies the proof-carrying structure that allows a verifier to check claims against the inductive definition.
\end{xmpl}

\begin{rmrk}
Example~\ref{xmpl.pow} connects directly to agent certification. 
Columns are not isolated output values; they are structured witnesses which reference each other to prove their validity. 
In an agentic setting this distinction separates a validator that checks whether an action is compliant from one that checks whether an action carries evidence of its compliance (in this case, columns pointing to one another to prove a recursive computation). 
This is precisely the proof-carrying structure that the certificates of Section~\ref{sect.how} would cryptographically attest to.
\end{rmrk}

\begin{rmrk}
\label{rmrk.pow.mitigate}
A more efficient specification of $a^b$ is: 
$$
\begin{array}{r@{\ }l}
\f{pow}(a,0)=&1
\\
\f{pow}(a,2*b)=&\f{pow}(a,b)*\f{pow}(a,b)
\\
\f{pow}(a,2*b+1)=&a*\f{pow}(a,b)*\f{pow}(a,b)
\end{array}
$$
This yields smaller derivations, which correspond algorithmically to a runtime that is logarithmic in $b$ rather than linear.
We could translate this into our logic, and it would just yield an implementation that permits shorter matrices.
\end{rmrk}

\begin{rmrk}
The logic in Figure~\ref{fig.syntax} is more expressive than it might seem.
In particular, negation \emph{is} expressible in the framework, via $\tf{reify}$.
A fuller treatment is in~\cite{gabbay2025arithmetisinglogic}.
A proposal for a SNARK-friendly constraint system is in~\cite{gabbay2026snarks}.
\end{rmrk}

\section{How to apply this mathematics to AI agents}
\label{sect.how}

Before we address the question of how to apply this mathematics to AI agents --- the motivating application of this paper, although the framework in Sections~\ref{sect.slice.of.maths} and~\ref{sect.examples} is general --- it is important to understand that polynomial semantics can be plugged directly into cryptographic proof schemes.

The semantics in Figure~\ref{fig.polynomial.semantics} is optimised for exposition, but variations optimised for implementation can be constructed.\footnote{One such variant has been submitted for publication, but at time of writing, submission rules prohibit sharing the results.  By time of publication I hope to be able to provide links.}

Thus it is important to appreciate that the presentation in Section~\ref{sect.slice.of.maths} is optimised for fitting onto one page including Figure~\ref{fig.polynomial.semantics} and Theorem~\ref{thrm.basic.obs}, consistent with clear exposition.
We can optimise for practical implementability instead, but that would be a longer and less accessible presentation: the passage from cleanly presentable theory to applicability is analogous to the passage from a high-level programming language to something more like assembly language or bytecode.
Be that as it may, the strategy for applying these ideas is simple: 
\begin{itemize*}
\item
Write down a predicate that expresses the desired notion of correctness, and translate it (using techniques analogous to program compilation) to an algebraic relation accepted by a cryptographic proof back end.  
Suitable back ends include integer or mixed-characteristic systems such as Zinc~\cite{garreta2025zinc} and polynomial-commitment/SNARK stacks such as Halo2 and related foundations~\cite{zcash2026halo2,boneh2021haloinfinite}.
\item 
Require the agent, or associated prover, to produce a cryptographic certificate accepted by the back end. 
Under the proof system's assumptions, and subject to stated soundness error, a valid certificate convinces the verifier that the encoded relation has a witness.
\end{itemize*}
If this seems conceptually simple, that is because it is.

Though this is a short paper, we can make this more concrete.
At the boundary of an agentic system, we can represent a certificate of validity abstractly as
\[
  \mathsf{Cert}
  =
  (\mathsf{policyID},\mathsf{action},\mathsf{pub},
   \mathsf{vk},\mathsf{paramsHash},\pi).
\]
Here \(\mathsf{policyID}\) identifies the formal predicate
\(\phi\) and the compiler version; \(\mathsf{action}\) is the
proposed action; \(\mathsf{pub}\) is the public instance data
needed to check that action; \(\mathsf{vk}\) and
\(\mathsf{paramsHash}\) identify the approved verifier key and
proof-system parameters; and \(\pi\) is the succinct cryptographic
proof. The intended statement verified by \(\pi\) is
\[
  \exists w.\; R_{\phi,\mathsf{action}}(\mathsf{pub},w)=0,
\]
where \(w\) is private witness data and
\(R_{\phi,\mathsf{action}}\) is the algebraic relation obtained
by compiling the formal predicate to polynomial constraints.

Thus verification of an action has the following form.

\begin{enumerate}
\item The policy author or system operator fixes a predicate
      \(\phi\) expressing the required correctness, safety, or
      compliance condition.
\item A compiler translates \(\phi\), together with the relevant
      public action data, into an algebraic relation
      \(R_{\phi,\mathsf{action}}\).
\item The agent, or an associated prover, produces a proof \(\pi\)
      that there exists witness data \(w\) satisfying that relation.
\item The receiving system checks \(\pi\) using the approved verifier
      key and rejects action if verification fails.
\end{enumerate}
We can sum this up as a slogan: \emph{logical validity translates to vanishing of the corresponding algebraic relation.}

In this interface the verifier need not trust the agent, inspect 
 internal implementation, or re-execute computation.
It just checks the proof. 
Conversely, the certificate proves only the compiled formal claim: it does not prove that the policy was wisely
chosen, that the compiler is bug-free, or that the proof-system parameters were correctly governed.

A practical compiler would use implementation-oriented arithmetisation techniques; for example multilinear encodings, range checks, and lookup arguments, targeting the back ends mentioned above (Zinc; Halo2).
The mathematics of Sections~\ref{sect.slice.of.maths} and~\ref{sect.examples} specifies the correctness condition that such a compiler must preserve. 

\paragraph{Roadmap.}
So we can propose a \emph{provisional roadmap.}
If a policy can be logically expressed as a predicate, and if a trusted compilation chain maps that predicate to an algebraic relation while preserving the zero-versus-positive philosophy of Figure~\ref{fig.zero-true-positive-false}, then any sound succinct argument for that relation yields a succinct certificate that the policy predicate is satisfied by the supplied witness data. 
Remaining questions are just what one would expect, moving from mathematics to a deployed assurance mechanism: how to author and audit correctness predicates; how to verify or test the compiler; how to select and parameterise the proof back end; how to benchmark proving cost; and how to record policy versions and verifier configurations for later accountability.

\section{Related work}

\paragraph{Formal specification and verification.}

Formal methods provide mature tools for specifying and proving properties of software and systems.
Hoare logic~\cite{hoare1969axiomatic}, proof-carrying code~\cite{necula1997pcc}, verified compilers like CompCert~\cite{leroy2009compcert}, and verified kernels like seL4~\cite{klein2009sel4} can provide machine-checkable assurance for critical systems.
\emph{Correctness-by-construction} makes assurance part of design and implementation~\cite{hall2002correctness}.
These approaches share with our proposal a formally specified notion of `correctness'.

However, the deployments differ: proof-carrying code is not a black-box technique, whereas our proposal imposes no restrictions on how an agent is designed and implemented, other than insisting that some or all of the actions that it proposes be associated to a succinct certificate proving validity of an agreed correctness predicate.

\paragraph{Runtime verification, proof-carrying code, and certificates.}

\emph{Runtime verification} checks traces of an executing system against formal properties~\cite{leucker2009runtime}.  
\emph{Proof-carrying code} asks an untrusted code producer to supply a proof that code satisfies a consumer's safety policy~\cite{necula1997pcc}.

Our proposal can be read as a cryptographic, succinct, and optionally zero-knowledge analogue of proof-carrying code: the consumer publishes or fixes a policy; the agent supplies an action and a certificate; and the verifier checks the certificate without re-running the agent's entire computation or inspecting all private witness data.

\paragraph{Succinct cryptographic proofs and verifiable computation.}

\emph{zk-SNARKs} and related succinct arguments allow a \emph{prover} to convince a \emph{verifier} that a relation holds between a \emph{witness} and an \emph{instance}, while keeping verifier effort small and, in zero-knowledge variants, hiding some or all of the witness.
Examples include Pinocchio~\cite{parno2013pinocchio}, SNARKs for C~\cite{bensasson2013snarksforc}, Groth's pairing-based arguments~\cite{groth2016size}, and later general frameworks for polynomial IOPs, constraint systems, and related lookup machinery~\cite{bensasson2016iop,bunz2020darkcompilers,thaler2022proofs,setty2023ccs,setty2024lasso}.  
The technical constructions underlying this paper use similar technology; in unpublished work we have shown how to arithmetise the first-order logic of this paper to polynomial constraints on integer/mixed-characteristic proof systems such as Zinc~\cite{garreta2025zinc}, and we are studying how to translate to Halo2~\cite{zcash2026halo2} so that logical specifications can be connected to cryptographic back ends.

\paragraph{zkVMs and proof-oriented execution environments.}

An engineering response to the difficulty of writing polynomial constraints directly is zero-knowledge virtual machines, like Cairo~\cite{goldberg2021cairo}, RISC0~\cite{risczero2026zkvm}, and Jolt~\cite{arun2024jolt}.
These all offer verifiable receipts of correct code execution.
Work on formally verifying Cairo execution in Lean reflects the importance of connecting cryptographic proof systems to proof-assistant-level assurance~\cite{avigad2021verifiedcairo}.

In contrast, the present proposal does not introduce another machine model.
Its high-level specification language is first-order logic itself.
Programs, policies, transaction rules, and compliance requirements are represented as predicates, which are compiled directly to polynomial proof obligations.

\paragraph{Choreographic programming and distributed correctness.}

Choreographic programming and multiparty session types address distributed correctness from complementary angles.
Multiparty session types use global protocol descriptions whose projections type-check communicating processes and establish communication safety, progress, and fidelity~\cite{honda2016multiparty}.
Choreographic programming writes a global coordination plan from which decentralised implementations can be generated correct by construction~\cite{cruzfilipe2017core} (subsequently developed to include functional~\cite{cruzfilipe2022functional} and machine-checked~\cite{cruzfilipe2023formal} formulations).
These ideas may be relevant to agentic AI systems.  
Our proposal does not replace choreographic or protocol-based design, and it deals with the orthogonal problem of how to ensure that possibly Byzantine, possibly black-box agents --- i.e. agents which we did not implement and which we may not be able to look inside --- should still carry cryptographic certificates of correct behaviour.

\section{Conclusions}

Formal methods provide the language of formal specification; correctness-by-construction provides a design discipline; runtime verification and proof-carrying code provide certificate-oriented assurance; and zk-SNARKs/zkVMs provide cryptographically checkable execution.

This paper articulates a synthesis of these ideas for trustworthy agentic AI: a consequential action may be logged and signed --- \emph{and also}, where feasible, it may be accompanied by a succinct certificate of validity of a formally specified notion of correctness, safety, or compliance.

It remains to render these ideas into cryptographically implementable schemes, e.g. by reduction to SNARK-friendly constraint systems.
This is ongoing work: a proposal has been very recently presented~\cite{gabbay2026snarks} and a conference paper submitted. 

Future work includes generating auditable deployment patterns; e.g. specification templates for common agent actions, benchmarks across proof back ends, recording of policy versions and verifier configurations, and independent cryptographic, formal-methods, and policy review.  

These are not mere engineering details. 
A cryptographic certificate proves satisfaction of the encoded formal predicate under the relevant proof-system assumptions and soundness bounds.
It does \emph{not} prove that the correctness predicate itself reflects safety, legal, operational, or ethical requirements. 
This specification gap must be addressed through governance, visibility, and risk-management processes; it cannot be closed by mathematics, logic or cryptography alone~\cite{chan2024visibility,nist2023airmf}.

\hyphenation{Mathe-ma-ti-sche}
\providecommand{\href}[2]{#2}
\providecommand{\nolinkurl}[1]{#1}
\providecommand{\bibdoi}[1]{\href{https://doi.org/#1}{\nolinkurl{doi:#1}}}
\providecommand{\biburl}[1]{\href{#1}{\nolinkurl{#1}}}

\end{document}